\documentclass[a4paper,11pt]{amsart}
\usepackage{amssymb,amsmath,amsfonts,amsthm}
\usepackage[english]{babel}
\usepackage[yyyymmdd]{datetime}
\usepackage{enumerate,expdlist}
\usepackage{tikz}
\usetikzlibrary{patterns}
\newcommand{\euler}{\mathrm{e}}
\newcommand{\RR}{\mathbb{R}}
\newcommand{\CC}{\mathbb{C}}
\newcommand{\NN}{\mathbb{N}}    
\newcommand{\ZZ}{\mathbb{Z}}
\newcommand{\PP}{\mathbb{P}}
\newcommand{\EE}{\mathbb{E}}

\newcommand{\vol}{\operatorname{vol}}
\newcommand{\cF}{\mathcal{F}}
\newcommand{\cV}{\mathcal{V}}
\newcommand{\cE}{\mathcal{E}}
\newcommand{\cD}{\mathcal{D}}
\newcommand{\cC}{\mathcal{C}}
\newtheorem{theorem}{Theorem}[section]

\newtheorem{proposition}[theorem]{Proposition}
\newtheorem{corollary}[theorem]{Corollary}

\theoremstyle{definition}
\newtheorem{definition}[theorem]{Definition}
\newtheorem{remark}[theorem]{Remark}
\newtheorem{example}[theorem]{Example}
\usepackage{microtype}
\usepackage{ellipsis}
\usepackage{typearea}
\usepackage[textsize=tiny,shadow]{todonotes}
\begin{document}
\title[UCP and their absence for Schr\"odinger eigenfunctions]{Unique continuation principles and their absence for Schr\"odinger eigenfunctions on combinatorial and quantum graphs and in continuum space}
\author{Norbert Peyerimhoff, Matthias T\"aufer, Ivan Veseli\'c}
\address[NP]{Department of Mathematical Sciences, Durham University, UK}
\address[MT \& IV]{Fakult\"at f\"ur Mathematik, Technische Universit\"at Dortmund, Germany}
\date{}
\thanks{\textcopyright 2017 by the authors. Faithful reproduction of this article, in its entirety, by any means is permitted for non-commercial purposes. \today, \jobname.tex}

\begin{abstract}
For the analysis of the Schr\"odinger and related equations it is of central importance whether a unique continuation principle (UCP) holds or not.
In continuum Euclidean space quantitative forms of unique continuation imply Wegner estimates and regularity properties of the integrated density of states (IDS) of Schr\"odinger operators with random potentials.
For discrete Schr\"odinger equations on the lattice only a weak analog of the UCP holds, but it is sufficient to guarantee the continuity of the IDS. For other combinatorial graphs this is no longer true.
Similarly, for quantum graphs the UCP does not hold in general and consequently, the IDS does not need to be continuous.
\end{abstract}

\maketitle

\section{Introduction}

Unique continuation properties for various function classes have been studied for many years.
They are of great importance when addressing uniqueness of solutions of partial differential equations, the propagation or regularity of solutions, and their growth behaviour.
More recently, they have been successfully applied in the spectral theory of random Schr\"odinger operators,
for instance to prove Wegner estimates and establish regularity properties of the integrated density of states (IDS).
\par
On the other hand it is well-known that for discrete Schr\"odinger operators on the lattice $\ZZ^d$
the analogue of the UCP does not hold.
This poses a serious difficulty for the analysis of discrete Schr\"odinger operators.
This is exemplified by the fact that there is still no proof of localisation for the multidimensional Anderson model with Bernoulli disorder while this has been established for the seemingly more difficult analogous problem in continuum space in \cite{BourgainK-05}.
Nevertheless, a certain weaker version of unique continuation, namely non-existence of finitely supported eigenfunctions, allows at least to conclude that the IDS of discrete Schr\"odinger operators on $\ZZ^d$ is continuous.
This, however, uses specific properties of the underlying combinatorial graph $\ZZ^d$ and does not need to be true for Laplace or Schr\"odinger operators on other graphs.
A prominent example for this phenomenon is the Laplace operator on a subgraph of $\ZZ^d$, generated by (random) percolation.
Another example is the discrete Laplacian on the Kagome lattice which is a planar graph exhibiting eigenfunctions with finite support.
In both examples finitely supported eigenfunctions lead to jumps of the IDS. The two properties are actually in a sense equivalent.
However, there is a condition on planar graphs, namely non-positivity of the so-called corner curvature, which excludes the existence of finitely supported eigenfunctions.
\par
For quantum graphs, more precisely, for Schr\"odinger operators on metric graphs, the UCP does not hold in general as well.
On the one hand, this can be understood as a consequence of the phenomenon encountered for planar graphs, since there is a way to ``translate`` spectral properties of equilateral quantum graphs to spectral properties of the underlying combinatorial graph.
On the other hand, as soon as the underlying combinatorial graph has cycles, the Laplacion on the corresponding equilateral quantum graph carries compactly supported, so-called \emph{Dirichlet} eigenfunctions on these cycles which can again lead to jumps in the IDS.
\par
The paper is structured as follows:
In Section~\ref{sect:Continuous} we discuss unique continuation principles for Schr\"odinger equations on subsets of $\RR^d$.
Then, in Section~\ref{sect:Z^d}, we turn to analogous discrete equations on the Euclidian lattice graph $\ZZ^d$, where we present both positive and negative results concerning unique continuation.
Section~\ref{sect:percolation} is devoted to subgraphs of the Euclidian lattice $\ZZ^d$, generated by percolation, i.e. by random removing vertices.
There, finitely supported eigenfunctions exist leading to jumps in the IDS.
After that, in Section~\ref{sect:planar_graphs}, we introduce the Kagome lattice as an example of a planar graph which exhibits finitely supported eigenfunctions and then present a combinatorial curvature condition which can ensure the non-existence of such finitely supported eigenfunctions.
The final Section~\ref{sect:quantum_graphs} is devoted to quantum graphs.
We explain how properties from the underlying combinatorial graph translate to the lattice graph and study the IDS.

\section{Unique continuation for solutions in continuum space}
\label{sect:Continuous}
Throughout this article we will use the following notation:
A measurable function $f$ on a domain $A \subset \RR^d$ is in $L^p (A)$, if $\lVert f \rVert_{L^p(A)} = \lVert f \rVert_p < \infty$, where $\lVert f \rVert_p = ( \int_A \lvert f \rvert^p)^{1/p}$ if $1 \leq p < \infty$ and $\lVert f \rVert_\infty = \operatorname{essup}_A \lvert f \rvert$, the essential supremum with respect to the Lebesgue measure.
If $B \subset A$, we write $\lVert f \rVert_{L^p(B)} = \lVert \chi_B f \rVert_{L^p(A)}$, where $\chi_B$ is the characteristic function of the set $B$, i.e. $\chi_B(x) = 1$ if $x \in B$ and $0$ else.
The function $f$ is said to be in $H^{k,p}$, $k \in \NN$, if $f$ and all weak derivatives of $f$ up to $k$-th order are in $L^p$.
For a vector $x \in \RR^d$ we will denote by $\lvert x \rvert = (x_1^2 + \dots x_d^2)^{1/2}$ its Euclidian norm.
The (open) ball of radius $r > 0$ around $x \in \RR^d$ is denoted by $B(x,r) = \{ y \in \RR^d \colon \lvert x - y \rvert < r \}$.
Furthermore, for $L > 0$ and $x \in \RR^d$, we will call $\Lambda_L(x) = x + (-L/2, L/2)^d \subset \RR^d$ the $d$-dimensional cube of sidelength $L$, centered at $x$. If $x = 0$, we will simply write $\Lambda_L$.
\begin{definition}
\label{def:UC}
A class of functions $\mathcal{F}$ on a connected domain $A \subset \RR^d$ has the unique continuation property (UCP), if for every nonempty and open $U \subset A$ every $f \in \mathcal{F}$ vanishing on $U$ must vanish everywhere.
If every eigenfunction of a partial differential operator $D$ has the UCP then we say that the operator $D$ has the UCP.
\end{definition}
Standard examples of operators having the UCP include the Laplace operator $\Delta$ or elliptic operators with analytic coefficients.
A breakthrough result was due to Carleman~\cite{Carleman-39} in 1939 who proved that $- \Delta + V$ with $V \in L_{\mathrm{loc}}^\infty$ has the UCP by using inequalities which are nowadays refered to as Carleman estimates.
We shall first have a look at some unique continuation properties which at first sight are weaker than the above definition.
In order to illustrate the mechanism how Carleman estimates imply unique continuation let us recall a proof of the following result, see~\cite{KenigRS-87}.
\begin{proposition}[Unique continuation from a half space, \cite{KenigRS-87}]
 \label{prop:UC_half_space}
 Let $d \geq 3$, $p = 2d/(d+2)$ and $V \in L^{d/2}(\RR^d)$.
 Then every $u \in H^{2,p}(\RR^d)$ satisfying $\lvert \Delta u \rvert \leq \lvert V u \rvert$ which vanishes on a half space must vanish everywhere.
\end{proposition}
In fact, we are going to show a slightly stronger statement.
By an infinite slab of width $\epsilon$, we denote a set $S \subset \RR^d$ which is a translation and rotation of
\[
  \{ x \in \RR^d \colon 0 < x_1 < \epsilon, x_2, \dots x_d \in \RR \}.
\]
In dimension $d = 2$, an infinte slab would be an infinite strip.
\begin{proposition}[Unique continuation from a slab]
 \label{prop:UC_strip}
 Let $d \geq 3$, $p = 2d/(d+2)$ and $V \in L^{d/2}(\RR^d)$.
 Then every $u \in H^{2,p}(\RR^d)$ satisfying $\lvert \Delta u \rvert \leq \lvert V u \rvert$ which vanishes on a infinite slab of width $\epsilon > 0$ must vanish everywhere.
\end{proposition}
The proof relies on the following Carleman estimate, which can be found e.g. in \cite{KenigRS-87}.
\begin{theorem}
 \label{thm:Carleman}
 Let $d \geq 3$, $p = 2d/(d+2)$ and $q = 2d/(d-2)$. Then there is a constant $C > 0$ such that for all $\nu \in \RR^d$, all $\lambda \in \RR$ and all $u$ with $\euler^{\lambda \left\langle \nu, x \right\rangle} u \in H^{2,p}(\RR^d)$ we have
 \[
  \lVert \euler^{\lambda \left\langle \nu, x \right\rangle} u \rVert_{L^q(\RR^d)}
  \leq
  C
  \lVert \euler^{\lambda \left\langle \nu, x \right\rangle} \Delta u \rVert_{L^p(\RR^d)}.
 \]
\end{theorem}

\begin{proof}[Proof of Proposition~\ref{prop:UC_strip}]
We choose $\rho > 0$ such that $\lVert V \rVert_{L^{d/2}(S_\rho)} \leq 1/(2C)$ for all infinite slabs $S_\rho$ of width $\rho$ where $C$ is the constant from Theorem~\ref{thm:Carleman}.
 By translation and rotation, we may assume that $u$ vanishes on the slab $\{ x \in \RR^d : - \epsilon < x_1 < 0 \}$ and it suffices to show $u \equiv 0$ in $S_\rho := \{ x \in \RR^d: 0 < x_1 < \rho \}$.
 Let now $\chi \in C^\infty(\RR^d)$ such that $\chi \equiv 0$ if $x_1 < - \epsilon$ and $\chi \equiv 1$ if $x_1 > 0$.
 We estimate, using H\"older's inequality and $\lvert \Delta u \rvert \leq \lvert V u \rvert$ to obtain for all $\lambda > 0$
 \begin{align*}
  \lVert \euler^{- \lambda x_1} u \rVert_{L^q(S_\rho)}
  &\leq \lVert \euler^{- \lambda x_1} \chi u \rVert_{L^q(\RR^d)}\\
  &\leq C \lVert \euler^{- \lambda x_1} \Delta( \chi u ) \rVert_{L^p(\RR^d)}\\
  &\leq C \lVert \euler^{- \lambda x_1} \Delta u  \rVert_{L^p(S_\rho)}
        +
        C \lVert \euler^{- \lambda x_1} \Delta( \chi u ) \rVert_{L^p(\RR^d \backslash S_\rho)}\\
  &\leq C \lVert \euler^{- \lambda x_1} V u  \rVert_{L^p(S_\rho)}
        +
        C \euler^{- \lambda \rho} \lVert \Delta u \rVert_{L^p(\RR^d)}\\
  &\leq C
        \lVert V \rVert_{L^{d/2}(\RR^d)}
        \cdot
        \lVert \euler^{- \lambda x_1} u  \rVert_{L^q(S_\rho)}
        +
        C \euler^{- \lambda \rho} \lVert \Delta u \rVert_{L^p(\RR^d)}\\
  &\leq \frac{1}{2} \lVert \euler^{- \lambda x_1} u  \rVert_{L^q(S_\rho)}
        +
        C \euler^{- \lambda \rho} \lVert \Delta u \rVert_{L^p(\RR^d)},
 \end{align*}
 where $q$ is the exponent from Theorem~\ref{thm:Carleman}.
 Substracting the first summand on the right hand side and multiplying by $\euler^{\lambda \rho}$, one finds
 \[
  \lVert \euler^{\lambda (\rho - x_1)} u \rVert_{L^q(S_\rho)}
  \leq
  2 C \lVert \Delta u \rVert_{L^p(\RR^d)}
 \]
 for all $\lambda > 0$.
 This is only possible if $u \equiv 0$ in $S_\rho$.
\end{proof}
Now one is in the position to conclude unique continuation properties of other domains.
\begin{proposition}[Outside-in and inside-out unique continuation, \cite{KenigRS-87}]
 Let $u \in H^{2,p}(\RR^d)$ satisfy $\lvert \Delta u \rvert \leq \lvert V u \rvert$ for a $V \in L^{d/2}(\RR^d)$.
 \begin{enumerate}[i)]
  \item If $u$ vanishes outside of an open ball of radius $\rho > 0$, it must vanish everywhere.
  \item If $u$ vanishes on an open ball of radius $\rho > 0$, it must vanish everywhere.
 \end{enumerate}
\end{proposition}
 Part i) is a special case of Proposition~\ref{prop:UC_half_space}, while the proof of Part ii) is based upon the transformation $u(x) \mapsto \tilde u(x) := u(x/ \lvert x \rvert^2) \cdot \lvert x \rvert^{-(d-2)}$.
\par
 So far, we found that eigenfunctions vanishing on half-spaces, slabs, outside and inside of balls must vanish everywhere.
 In particular, the latter implies the notion of unique continuation as in Definition~\ref{def:UC}.
 The assumption $V \in L^{d/2}(\RR^d)$ can be substantially relaxed, but we are not going to focus our attention on this issue and refer to the references~\cite{Wolff-92b,Wolff-95,KochT-01}.
 We emphasize, however that we exploited rotational symmetry and the transformation $x \mapsto x/\lvert x \rvert^2$.
 On the lattice $\ZZ^d$ this is not going to work any more.
\par
 While unique continuation itself has turned out to be a useful tool for many applications~\cite{JerisonK-85, EscauriazaKPV-12}, in some situations, more information is required.
 We speak of \emph{Quantitative unique continuation} if a function which is ``small'' on $U$ cannot be ``too large'' on the whole domain $A$.
 Of course the notion of smallness needs some clarification.
 It can be formulated in terms of different norms, local maxima, etc. and there is a connection to vanishing speed of functions in a neighbourhood of their zero set.
 We are going to cite some cases of quantitative unique continuation principles and some resulting applications.
\par
The first example concerns vanishing speed of solutions of the Laplace-Beltrami operator on compact manifolds with the explicit dependence $\euler^{\sqrt{E}}$ on the eigenvalue - a term that we will encounter later on.
It is due to \cite{DonnellyF-88} and follows by combining Thm.~4.2 (i) with the second displayed formula on p.~174 in~\cite{DonnellyF-88}.

\begin{theorem}
 Let $M$ be a closed, compact $C^\infty$ Riemannian manifold.
 Then there are constants $C_1, C_2 \geq 0$ such that for every $u \not \equiv 0$ and $- \Delta u = E u$ and every $x_0 \in M$, we have
 \[
  \epsilon^{C_1 + C_2 \sqrt{E}}
  \cdot
  \max_{x \in M} \lvert u(x) \rvert
  \leq
  \max_{x \in B(x_0,\epsilon)} \lvert u(x) \rvert
  \quad
  \text{for small enough}\
  \epsilon > 0
 \]
 i.e. $u$ can at most vanish of order $C_1 + C_2 \sqrt{E}$.
\end{theorem}
In particular, if an eigenfunction $u$ of the Laplace-Beltrami operator is zero in a non-empty open set, it certainly vanishes of infinite order and thus $u \equiv 0$, i.e. it has the UCP.
In \cite{Bakri-13}, similar results were proven for a larger class of second order differential operators which allowed for a potential and first order terms.

Now we turn to vanishing properties at infinity.
In this setting, one wants to understand the fastest possible rate at which a function can decay as the norm of its argument tends to infinity.

\begin{theorem}[Quantitative UCP for eigenfunctions of Schr\"odinger operator, \cite{BourgainK-05}]
\label{thm:ucp_Bourgain_Kenig}
 Assume $\Delta u = V u + \gamma$ in $\RR^d$, $u(0) = 1$, $\lvert u \rvert \leq C$ and $\lVert V \rVert_\infty \leq C$.
 Then there are $C_1, C_2 > 0$ such that for every $x_0 \in \RR^d$, we have
 \begin{equation}
 \label{eq:ucp_Bourgain_Kenig_1}
  \max_{\lvert x - x_0 \rvert \leq 1} \lvert u(x) \rvert
  +
  \lVert \gamma \rVert_\infty
  >
  C_1
  \exp
   \left(
   - C_2 (\log \lvert x_0 \rvert )
   \lvert x_0 \rvert^{4/3}
   \right).
 \end{equation}
\end{theorem}

Theorem~\ref{thm:ucp_Bourgain_Kenig} was an essential ingredient in proving \emph{spectral localization}, i.e. almost sure occurrence of dense pure point spectrum with exponentially decaying eigenfunctions for the \emph{Anderson-Bernoulli model}
\[
 H_\omega = - \Delta + V_\omega,
 \quad
 V(x) = \sum_{j \in \ZZ^d} \omega_j u(x-j)
\]
where $\omega_j$ are independent and identically distributed Bernoulli random variables (i.e. they are either $0$ or $1$) and $\phi$ is a smooth, positive, compactly supported single-site potential.

While localization has been well established before in the case of the $\omega_j$ having an absolutely continuous (with respect to the Lebesgue measure) probability measure,
see~e.g.~\cite{Stollmann-01,GerminetK-13b},
the case of Bernoulli distributed random variables had been more challenging and Theorem~\ref{thm:ucp_Bourgain_Kenig} turned out to be an essential ingredient in the proof.
In fact, since there is no lattice analogue of Theorem~\ref{thm:ucp_Bourgain_Kenig},
the question of localisation for the Anderson-Bernoulli model on the lattice $\ZZ^d$ is still open, except in the case of dimension $d = 1$ wehere different methods are available,
see \cite{CarmonaKM-87}, Theorem~2.1.
\par
In order to formulate the next result, we need to define the density of states (DOS) and the integrated density of states (IDS).
Let $V \in L^\infty(\RR^d)$ and $H = - \Delta + V$ on $L^2(\RR^d)$.
For a $d$-dimensional cube $\Lambda$, we call $H_\Lambda$ the restriction of $H$ to $L^2(\Lambda)$ with Dirichlet boundary conditions (i.e. by prescribing the value $0$ at the boundary of $\Lambda$).
Its spectrum consists of an increasing sequence of eigenvalues of finite multiplicity with the only accumulation point at $+ \infty$.
The finite volume density of states measure $\eta_\Lambda$ is defined by
\[
 \eta_\Lambda (B) := \frac{1}{\lvert \Lambda \rvert} \sharp \{ \text{Eigenvalues of}\ H_\Lambda \ \text{in}\ B \}
\]
for any Borel set $B \subset \RR$. Here and in the sequel we count eigenvalues according to their multiplicity.
If the potential $V$ is periodic, the density of states measure can be defined as the limit
\[
 \eta(B)
 :=
 \lim_{L \to \infty} \eta_{\Lambda_L} (B).
\]
More generally, if we have an ergodic random family $\{ V_\omega \}_{\omega \in \Omega}$ of potentials, there is convergence of the integrated density of states to a non-random function
\[
 N(E)
 :=
 \lim_{L \to \infty}
 \eta_{\Lambda_Lf}((- \infty, E])
\]
for almost every $E$ and almost every $\omega \in \Omega$.
For generic Schr\"odinger operators, $\eta$ might not be well-defined but one can still define the density of states outer-measure as
\[
 \eta^\ast(B)
 :=
 \limsup_{L \to \infty} \sup_{x \in \RR^d} \eta_{\Lambda_L(x)} (B),
\]

In \cite{BourgainK-13}, a version of Theorem~\ref{thm:ucp_Bourgain_Kenig} was applied to prove continuity of the density of states (outer-)measure in dimension $d = 2,3$.
The case of dimension $d = 1$ had already been proved in~\cite{CraigS-83b}.

\begin{theorem}[\cite{BourgainK-13}]
\label{thm:BourgainKlein}
 Let $H = - \Delta + V$ be a Schr\"odinger operator with bounded potential $V$ and let the dimension $d \in \{1,2,3\}$.
 Then for every $E_0 \in \RR$ there are constants $C_1, C_2$, depending only on $E_0$, $\lVert V \rVert_\infty$ and $d$ such that for every $E \leq E_0$ and every small enough $\epsilon$
 \[
  \eta^\ast( [E, E + \epsilon])
  \leq
  \frac{C_1}{( \log 1/\epsilon)^{C_2}},
 \]
 i.e. the density of states outer-measure is continuous.
\end{theorem}
If  $d = 1$, one can choose $C_2=1$ and $\epsilon \in(0,1/2)$ cf.~Theorem 5.1 in ~\cite{CraigS-83b}.
The restriction to dimension $d \leq 3$ is due to the exponent $4/3$ in (an analogue of) ineq.~\eqref{eq:ucp_Bourgain_Kenig_1} which originates from the particular Carleman inequality they use.
In fact, if this exponent was to be replaced by $\beta > 1$, then Theorem~\ref{thm:BourgainKlein} would hold for all dimensions $d < \beta/(\beta - 1)$, whence it is desirable to reduce the exponent $4/3$ in ineq.~\eqref{eq:ucp_Bourgain_Kenig_1} to $1$.
However there is a classic example \cite{Meshkov-92} which shows that this will not be feasible using Carleman estimates,
whence new approaches to unique continuation will be required in order to lift the proof of Theorem~\ref{thm:BourgainKlein} to higher dimensions.
\par
We will now study \emph{scale-free} unique continuation, i.e. we will study quantitative unique continuation results which hold uniformly over a large number of scales and geometric settings.
For that purpose, we introduce the following definition:
\begin{definition}
 Let $0 < \delta < 1/2$.
 We say that a sequence $Z = \{z_j\}_{j \in \ZZ^d}$ is $\delta$-\emph{equidistributed}, if for every $j \in \ZZ^d$ we have $B(z_j,\delta) \subset j + \Lambda_1$.
 Corresponding to a $\delta$-equidistributed sequence and $L > 0$, we define
 \[
  S_\delta(L) := \bigcup_{j \in \ZZ^d} B(z_j,\delta) \cap \Lambda_L.
 \]
\end{definition}
The simplest example of a $\delta$-equidistributed set would be $\ZZ^d$ itself.

\begin{figure}
 \caption{Examples of $S_\delta(5)$ for different $\delta$-equidistributed arrangements}
  \begin{tikzpicture}[scale = 0.8]
%  non-periodic
  \begin{scope}[xshift = 0cm]
%   Grid
  \draw[black!50] (-.25, -.25) grid (5.25, 5.25);
  \foreach \x/\y in {0.5/0.5,1.2/0.2,2.6/0.6,3.4/0.3,4.7/0.8,
                                                0.2/1.2,1.3/1.6,2.4/1.5,3.7/1.7,4.5/1.3,
                                                0.3/2.4,1.2/2.3,2.6/2.7,3.7/2.5,4.5/2.4,
                                                0.6/3.3,1.5/3.7,2.3/3.3,3.4/3.5,4.7/3.7,
                                                0.4/4.7,1.5/4.5,2.3/4.6,3.7/4.6,4.5/4.3
   }{
   \filldraw[black!75] (\x,\y) circle (5pt);
   \draw (\x,\y) circle (5pt);}
  \end{scope}
%  periodic
  \begin{scope}[xshift = 7cm]
%   Grid
   \draw[black!50] (-.25, -.25) grid (5.25, 5.25);
%   Dots
    \foreach \x in {0,...,4}{
     \foreach \y in {0,...,4}{
      \filldraw[black!75] (\x+0.5, \y+0.5) circle (5pt);
    }}
 \end{scope}
 \end{tikzpicture}
\end{figure}
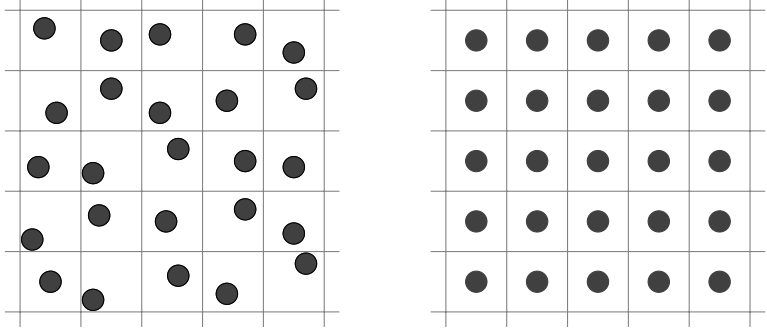

\begin{theorem}[Quantitative UCP for eigenfunctions, \cite{Rojas-MolinaV-13}]
 \label{thm:ucp_RojasMolinaVeselic}
 Fix $K_V \in [0, \infty)$ and $\delta \in (0,1/2)$.
 Then there is a constant $C > 0$ such that for all $L \in \NN_{\mathrm{odd}} = \{ 1, 3, \dots \}$, all measurable $V: \Lambda_L \to [- K_V, K_V]$ and all real-valued $\psi$ in the domain of the Laplace operator on $\Lambda_L$ with Dirichlet or periodic boundary condition satisfying
 \[
  \lvert \Delta \psi \rvert
  \leq
  \lvert V \psi \rvert
 \]
 we have
 \[
   \lVert \psi \rVert_{L^2(S_\delta(L))}^2
   \geq
   \left( \frac{\delta}{C} \right)^{C + C K_V^{2/3}}
   \lVert \psi \rVert_{L^2(\Lambda_L)}^2
   .
 \]
\end{theorem}

Theorem~\ref{thm:ucp_RojasMolinaVeselic} is called a \emph{scale-free} unique continuation principle because the constant on the right hand side does not depend on the scale $L$.
It has been used to study the spectrum of random Schr\"odinger operators, more precisely the \emph{Delone-Anderson model}
\begin{equation}
 \label{eq:Definition_Delone-Anderson}
 H_\omega
 =
 -\Delta + V_\omega,
 \quad
 V_\omega(x)
 =
 \sum_{y \in \mathcal{D}} \omega_y u(x - y)
\end{equation}
where $u$ is a compactly supported, positive and bounded function, the $\omega_y$ are independent and identically distributed, 
bounded random variables with a bounded density and $\mathcal{D} \subset \RR^d$ is a \emph{Delone set}.
The latter means that there are $0 < L_1 < L_2$ such that for all $x \in \RR^d$ we have $\sharp \{ y \in \mathcal{D} \cup \Lambda_{L_1} \} \leq 1$ and $\sharp \{ y \in \mathcal{D} \cup \Lambda_{L_2} \} \geq 1$.
Every $\delta$-equidistributed set is a Delone set and every Delone set can (after scaling) be decomposed into a $\delta$-equidistributed set and some remaining set, see e.g.~\cite{Rojas-MolinaV-13}.
In~\cite{Rojas-MolinaV-13}, Theorem~\ref{thm:ucp_RojasMolinaVeselic} was used to prove the following Wegner estimate:
\begin{theorem}
 \label{thm:Wegner_Delone_Anderson}
 Let $\{ H_\omega \}_{\omega \in \Omega}$ be a Delone-Anderson Hamiltonian as in~\eqref{eq:Definition_Delone-Anderson}.
 For every $E_0$ there is a constant $C_W$ such that for all $E \leq E_0$, all $\epsilon \leq 1/3$, all $L \in \NN_{odd}$ we have
 \begin{equation}
 \label{eq:Wegner_Delone_Anderson}
  \EE \left[ \sharp \{ \text{Eigenvalues of}\ H_{\omega,\Lambda} \ \text{in}\ [E - \epsilon, E + \epsilon] \} \right]
  \leq
  C_W
  \cdot
  \epsilon
  \cdot
  \lvert \ln \epsilon \rvert^d
  \cdot
  \lvert \Lambda_L \rvert.
 \end{equation}
\end{theorem}
Wegner estimates serve as an induction ancor in the multi-scale analysis, an inductive process which establishes localization, i.e. the almost sure occurrence of pure point spectrum with exponentially decaying eigenfunctions for $H_\omega$, at low energies.
Note that the right hand side in Ineq.~\eqref{eq:Wegner_Delone_Anderson} is $o(\epsilon^\theta)$ as $\epsilon \to 0$ for every $\theta \in (0,1)$.
Therefore, if the integrated density of states of $H_\omega$ exists, it will be (locally) H\"older continuous with respect to any exponent $\theta \in (0,1)$.
Since, however, the Delone-Anderson model is not necessarily ergodic, existence of its IDS is a delicate issue, see~\cite{GerminetMRM-15}.
\par
In \cite{Rojas-MolinaV-13}, the question had been raised if a similar statement as in Theorem~\ref{thm:Wegner_Delone_Anderson} holds uniformly all for finite linear combination of eigenfunctions with eigenvalues below a threshold $E_0$.
Such results had been known before, cf.~\cite{CombesHK-03}, albeit only in the special case where both the potential $V$ and the Delone set $D$ were $\ZZ^d$-periodic and without the explicit dependence on $\delta$ and $K_V$.
They had led to Lipshitz continuity of IDS in the usual alloy-type or continuum Anderson model, cf. \cite{CombesHK-07}.
However, the proof of these unique continuation principles had relied on Floquet theory which only allowed for the periodic setting and a compactness argument which yielded no information on the influence of the parameters $\delta$ and $K_V$.
A partially positive answer to the question raised in \cite{Rojas-MolinaV-13} was given in \cite{Klein-13} where Theorem~\ref{thm:ucp_RojasMolinaVeselic} was generalized to linear combinations of eigenfunctions with eigenvalues in a small energy interval.
This allowed to drop the $\ln \epsilon$ term in \eqref{eq:Wegner_Delone_Anderson}.
A full answer to the question raised in \cite{Rojas-MolinaV-13} was given by the following Theorem.

\begin{theorem}[\cite{NakicTTV-15,NakicTTV}]
 \label{thm:ucp_NTTV}
 There is $N = N(d)$ such that for all $\delta \in (0,1/2)$, all $\delta$-equidistributed sequences, all measurable and bounded $V: \RR^d \to \RR$, all $L \in \NN$, all $E_0 \geq 0$ and all $\phi \in \mathrm{Ran} (\chi_{(-\infty,E_0]}(H_L))$ we have
 \begin{equation}
 \label{eq:result1}
 \lVert \phi \rVert_{L^2 (S_\delta (L))}^2
 \geq \delta^{N \bigl(1 + \lVert V \rVert_\infty^{2/3} + \sqrt{E_0} \bigr)}
 \lVert \phi \rVert_{L^2 (\Lambda_L)}^2.
 \end{equation}
\end{theorem}
Theorem~\ref{thm:ucp_NTTV} was a missing ingredient for treating new models of random Schr\"odinger operators such as the \emph{standard breather model}:
Let $\{ \omega_j \}_{j \in \ZZ^d}$ be i.i.d. random variables on a probability space $(\Omega, \PP)$ which are distributed according to the uniform distribution on the interval $[0, 1/2]$ and define the standard breather potential
\[
 V_\omega (x)
 :=
 \sum_{j \in \ZZ^d} \chi_{B(j, \omega_j)} (x)
\]
where $\chi_{B(x, r)}$ denotes the characteristic function of a ball of radius $r$, centered at $x$.
Then, the standard breather model is the family of operators $- \Delta + V_\omega$, $\omega \in \Omega$ on $L^2(\RR^d)$ and can be seen as a prototype for a random Schr\"odinger operator where the random variables enter in a non-linear manner.
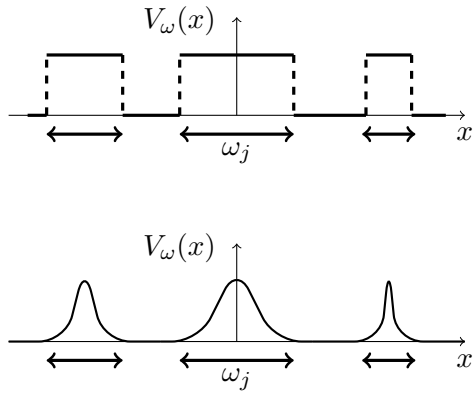
\begin{figure}
 \caption{Realizations of the standard breather potential and of a general random breather potential}
 \begin{tikzpicture}
  % Picture of standard breather potential
  \begin{scope}
  % grid
  \draw[->] (-3,0) -- (3,0);
  \draw[->] (0,0) -- (0,1.3);
  \draw (-.75,1.25) node {$V_\omega(x)$};
  \draw (3,-.25) node {$x$};
  % Potential
  \draw[very thick] (-2.75,0) -- (-2.5,0);

  \draw[very thick, dashed] (-2.5,0) -- (-2.5,.8);
  \draw[very thick] (-2.5,.8) -- (-1.5,.8);
  \draw[very thick, dashed] (-1.5,.8) -- (-1.5,0);

  \draw[very thick] (-1.5,0) -- (-.75,0);

  \draw[very thick, dashed] (-.75,0) -- (-.75,.8);
  \draw[very thick] (-.75,.8) -- (.75,.8);
  \draw[very thick, dashed] (.75,.8) -- (.75,0);

  \draw[very thick] (.75,0) -- (1.7,0);

  \draw[very thick, dashed] (1.7,0) -- (1.7,.8);
  \draw[very thick] (1.7,.8) -- (2.3,.8);
  \draw[very thick, dashed] (2.3,.8) -- (2.3,0);

  \draw[very thick] (2.3,0) -- (2.75,0);

% arrows for the \omega_j
  \draw[<->, very thick] (-.75,-.25) -- (.75,-.25);
  \draw (0,-.5) node {$\omega_j$};
  \draw[<->, very thick] (-2.5,-.25) -- (-1.5,-.25);
  \draw[<->, very thick] (1.65,-.25) -- (2.35,-.25);
  \end{scope}

  % Picture of general random breather potential
  \begin{scope}[yshift = - 3 cm]
  % grid
  \draw[->] (-3,0) -- (3,0);
  \draw[->] (0,0) -- (0,1.3);
  \draw (-.75,1.25) node {$V_\omega(x)$};
  \draw (3,-.25) node {$x$};
  % bumps
  \draw[rounded corners = 10 pt, thick] (-1,0) -- (-.5,0) -- (0,1) -- (.5,0) -- (1,0);
  \draw[rounded corners = 10 pt, thick] (-3,0) -- (-2.25,0) -- (-2,1) -- (-1.75,0) -- (-1,0);
  \draw[rounded corners = 10 pt, thick] (1,0) -- (1.9,0) -- (2,1) -- (2.1,0) -- (3,0);
  % arrows for the \omega_j
  \draw[<->, very thick] (-.75,-.25) -- (.75,-.25);
  \draw (0,-.5) node {$\omega_j$};
  \draw[<->, very thick] (-2.5,-.25) -- (-1.5,-.25);
  \draw[<->, very thick] (1.65,-.25) -- (2.35,-.25);

  \end{scope}

 \end{tikzpicture}
\end{figure}
\begin{theorem}[Wegner estimate for the standard breather model, \cite{NakicTTV-15, TaeuferV-15}]\label{thm:Wegner_Breather}
For every $E_0 \in \RR$ there are $C > 0$, $0 < \theta < 1$ such that for every $E < E_0$, every $L \in \NN$ and every small enough $\epsilon > 0$ we have
 \begin{equation}
 \label{eq:Wegner_Breather}
  \EE \left[ \sharp \{ \text{Eigenvalues of}\ H_{\omega,\Lambda_L} \ \text{in}\ [E - \epsilon, E + \epsilon] \} \right]
  \leq
  C
  \epsilon^\theta
  L^d.
 \end{equation}
\end{theorem}
This implies (non-uniform) H\"older continuity at $E$ of order $\theta$ of the corresponding IDS and can be used to establish localization for the standard breather model via multi-scale analysis.
\par
Actually, Theorem~\ref{thm:Wegner_Breather} holds in a much more general setting, see \cite{NakicTTV}.
We only mention here the (general) random breather model in which the characteristic functions of balls with random radii are replaced by random dilations of radially decreasing, compactly supported, bounded and positive function $u$
\[
 V_\omega(x)
 =
 \sum_{j \in \ZZ^d} u \left( \frac{x - j}{\omega_j} \right).
\]
Examples for $u$ are the smooth function
\[
   u(x) = \exp \left( - \frac{1}{1 - \lvert x \rvert^2} \right) \chi_{\lvert x \rvert < 1},
\]
or the hat potential
\[
   u(x) = \chi_{\lvert x \rvert < 1} (1 - \lvert x \rvert).
\]
Another application of Theorem~\ref{thm:ucp_NTTV} concerns decorrelation estimates and the spectral statistics of random Schr\"odinger operators in dimension 1, cf.~\cite{Shirley-15}.
\par
Theorem~\ref{thm:ucp_NTTV} can also be applied in the context of control theory for the heat equation to show null controllability for the heat equation.
More precisely, Theorem~\ref{thm:ucp_NTTV} can be used to give more explicit statements in the context of results obtained in~\cite{RousseauL-12}, cf.~\cite{NakicTTV}.

\section{Unique continuation problem for solutions on Euclidean lattice graphs}
\label{sect:Z^d}

\begin{definition}[Discrete Laplacian on $\ZZ^d$]
\label{def:Discrete_Lapliacian_ZZ^d}
 We define the \emph{discrete Laplacian} on functions $f: \ZZ^d \to \CC$ as
 \[
  (\Delta f)(i) = \sum_{i \sim j} (f(j) - f(i)) = \sum_{i \sim j} f(j) - 2 d \cdot f(i) ,
 \]
 where $i \sim j$ means that $i$ is a direct neighbour of $j$, i.e. $\lvert i - j \rvert = 1$.
\end{definition}

\begin{remark}[Why is this called ``Laplacian''?]
 If we think of $(f(i))_{i \in \ZZ^d}$ as evaluations of a function $f : \RR^d \to \CC$ on the points $i \in \ZZ^d$ and approximate the difference quotient $(f(x + \epsilon) - f(x))/\epsilon$ with $\epsilon = 1$, the minimal coarsness possible, we find
 \[
  f'(i + 1/2) \approx f(i+1) - f(i)
  \quad
  \text{and}
  \quad
   f'(i - 1/2) \approx f(i) - f(i-1)
 \]
 whence
 \[
  f''(i) \approx f'(i + 1/2) - f'(i - 1/2) \approx f(i-1) - 2 f(i) + f(i+1).
 \]
 In dimension $d$, this translates to
 \[
  (\Delta f) (i) \approx \sum_{i \sim j} f(j) - 2 d \cdot f(i)
 \]
%  Since adding $2d$ times the identity only leads to an energy shift, the latter is sometimes dropped for sake of simplicity and we study instead the \emph{adjacency matrix}
%  \[
%   (H_0 f) (i)
%   =
%   \sum_{j \sim i} f(j).
%  \]

\end{remark}

In the following examples we consider $-\Delta + V$ where $V : \ZZ^d \to \RR$.
\begin{example}[Unique Continuation from half spaces in $\ZZ^d$ with border parallel to an axis]
 Let $f : \ZZ^d \to \CC$ satisfy $(- \Delta  + V) f = 0$ on $\ZZ^d$ and $f(j) = 0$ for all $j = (j_1, ..., j_d) \in \ZZ^d$ with $j_1 \leq 0$.
 Let $i \in \ZZ^d$ with $i_1 = 0$.
 Then
 \[
  - \sum_{j \sim i} f(j) + (2 d + V(i)) f(i) = 0
 \]
 but the only unknown term is $f((1, i_2, \dots, i_d))$ and therefore must be zero.
 We see that $f$ must be zero on the slab $\{ j \in \ZZ^d : j_1 = 1 \}$.
 Inductively, we find $f \equiv 0$ in every slab of width $1$ whence $f = 0$ on $\ZZ^d$.
 By the very same argument we establish unique continuation from a slab
 $\{ (j_1, ..., j_d) \in \ZZ^d \colon j_1 = k\ \text{or}\ k+1 \}$ of width $2$.
\end{example}

\begin{figure}
\caption{Unique continuation from a half space and a double strip in dimension $2$}
\begin{tikzpicture}[scale = .8]
%%%%%%
%%% UC from a half space
%%%%%%
%  white dots
 \foreach \x in {2,...,8}
  \foreach \y in {1,...,4}{
  \draw[thick, black!60] (\x,\y) circle (.1cm);
  }
%  black dots
 \foreach \x in {2,...,4}
  \foreach \y in {1,...,4}{
  \draw[thick, fill = black] (\x,\y) circle (.1cm);
   }
 %  line separating half spaces
  \fill[pattern = north east lines, opacity = 0.4] (1.5,.5) rectangle (4.5,4.5);
  \draw[thick, black!50] (4.5,.5) -- (4.5,4.5);
 %  node marking black dots as ``known'' and white ones as ``unknown''
%   \draw (3,5.5) node {\color{black}$\psi$ known};
  \draw (3,0) node {\color{black}  $f = 0$};
%   \draw (10,3) node {$(\Delta + V) \psi = 0$};

  \draw (6.5,0) node {$f$ unknown};
 %  inequality with one unknown
  \draw[very thick, black] (3.2,2) -- (3.8,2);
  \draw[very thick, black] (4.8,2) -- (4.2,2);
  \draw[very thick, black] (4,1.2) -- (4,1.8);
  \draw[very thick, black] (4,2.8) -- (4,2.2);
%   \draw (9,1.5) node {\footnotesize \color{black!80} inequality with one unknown; solvable};
  \draw (6.5,1.5) node {\footnotesize \color{black!80} $f = 0$ here};
  \draw[thick, ->, rounded corners = 5pt, black!80] (5.5,1.5) -- (5,1.5) -- (5,1.8);
 %  Coloring one dot black
  \draw[thick, fill = black] (5,2) circle (.1cm);
%
%%%%%%
%%% UC from a double strip
%%%%%%
%
\begin{scope}[xshift = 10.5cm]
%  white dots
 \foreach \x in {0,...,7}
  \foreach \y in {1,...,4}{
  \draw[thick, black!60] (\x,\y) circle (.1cm);
  }
%  black dots
 \foreach \x in {3,...,4}
  \foreach \y in {1,...,4}{
  \draw[thick, fill = black] (\x,\y) circle (.1cm);
   }
 %  line separating half spaces
  \fill[pattern = north east lines, opacity = 0.4] (2.5,.5) rectangle (4.5,4.5);
  \draw[thick, black!50] (4.5,.5) -- (4.5,4.5);
  \draw[thick, black!50] (2.5,.5) -- (2.5,4.5);
 %  node marking black dots as ``known'' and white ones as ``unknown''
%   \draw (3,5.5) node {\color{black}$\psi$ known};
  \draw (3.5,0) node {\color{black}  $f = 0$};
%   \draw (10,3) node {$(\Delta + V) \psi = 0$};
  \draw (6, 0) node {$f$ unknown};
  \draw (1, 0) node {$f$ unknown};
%  inequality with one unknown
  \draw[very thick, black] (3.2,2) -- (3.8,2);
  \draw[very thick, black] (4.8,2) -- (4.2,2);
  \draw[very thick, black] (4,1.2) -- (4,1.8);
  \draw[very thick, black] (4,2.8) -- (4,2.2);
%  another inequality with one unknown
  \draw[very thick, black, dotted] (2.2,3) -- (2.8,3);
  \draw[very thick, black, dotted] (3.8,3) -- (3.2,3);
  \draw[very thick, black, dotted] (3,2.2) -- (3,2.8);
  \draw[very thick, black, dotted] (3,3.8) -- (3,3.2);
%   \draw (9,1.5) node {\footnotesize \color{black!80} inequality with one unknown; solvable};
  \draw (6.5,1.5) node {\footnotesize \color{black!80} $f = 0$ here};
  \draw[thick, ->, rounded corners = 5pt, black!80] (5.5,1.5) -- (5,1.5) -- (5,1.8);
  \draw (.5,2.5) node {\footnotesize \color{black!80} $f = 0$ here};
  \draw[thick, ->, rounded corners = 5pt, black!80] (1.5,2.5) -- (2,2.5) -- (2,2.8);
 %  Coloring one dot black
  \draw[thick, fill = black] (5,2) circle (.1cm);
  \draw[thick, fill = black] (2,3) circle (.1cm);
 \end{scope}
 \end{tikzpicture}
\end{figure}

\begin{example}[No unique continuation from a double slab where one point has been omitted]
 For the one omitted point, we can prescribe any value.
 Then, there is a unique continuation (by induction over infinite slabs of width $1$).
 Therefore, we have a $1$-dimensional family of possible continuations.
\end{example}

\begin{example}[No unique continuation from a double slab, where $n$ points have been omitted]
 We prescribe values for the $n$ points and find a unique continuation.
 Therefore, we have an $n$-dimensional family of possible continuations.
\end{example}

\begin{example}[No unique continuation from a half space with border in a $45^{\circ}$ angle to the axes]
\label{ex:half_space_45}
 For simplicity, we consider the case $d = 2$ and $V \equiv 0$.
%  ; the generalization to higher dimensions and arbitrary $V$ is straightforward and considering the adjacency matrix $H_0$ instead of the Laplacian $\Delta$ corresponds to adding the potential $V \equiv 2 d$..
 Let $f : \ZZ^2 \to \CC$ satisfy $\Delta f = 0$ on $\ZZ^2$ and $f \equiv 0$ on a diagonal half-space $\{(j_1,j_2) \in \ZZ^2 \colon j_1 + j_2 \leq 0\}$.
 This does not imply $f \equiv 0$ on $\ZZ^d$.
 In fact, as soon as a value of $f$ on an additional point in the anti-diagonal line  $\{(j_1, j_2) \colon j_1 + j_2 = 1 \}$ is given,
then the values on the whole anti-diagonal  can be recovered successively from the equations
 \[
  0 = \sum_{j \sim i} f(j) - 4 f(i)
 \]
 for $i$ in $\{ i \in \ZZ^2 \colon (i_1 + i_2) = 0 \}$, cf. Figure~\ref{fig:UC_Z_diagonal_half_space}.
 Inductively, we find that there is one degree of freedom in every infinite
 anti-diagonal  $\{(j_1, j_2) \colon j_1 + j_2 = k \}_{k \in \NN}$ and we found an infinite dimensional family of possible continuations.
 \end{example}

 This illustrates the difference to the $\RR^d$ case: While $\RR^d$ is invariant under rotations, $\ZZ^d$ is not whence some unique continuation properties only hold in certain directions.
 However, the next proposition shows that on $\ZZ^2$, the half-spaces with border in a $45^\circ$ angle to the axes are the only ones for which unique continuation fails.

\begin{proposition}[Unique continuation in $\ZZ^2$ from half spaces in almost all directions]
 Let $f: \ZZ^2 \to \CC$ satisfy $\Delta f = 0$ on $\ZZ^2$ and $f \equiv 0$ on a half-space $\{ j \in \ZZ^2 \colon \langle j, \nu \rangle \leq \alpha \}$ where $\nu$ is not parallel to $(1,1)$ or $(-1,1)$, i.e. the border of the half-space is not in a $45^\circ$ angle to an axis.
 Then $f \equiv 0$ on $\ZZ^d$.
\end{proposition}

\begin{proof}
 By symmetry between the coordinate axes and reflection, we may assume $\nu = (1, \lambda)$ where $\lambda \in [0,1)$.
 Similar considerations as in Example~\ref{ex:half_space_45} show that $u$ will vanish on the anti-diagonal line $\{j \in \ZZ^2 \colon j_1 + j_2 = c\}$ as soon as $u$ vanishes on a set $Q_{c_1, c_2} := \{ j \in \ZZ^2 \colon j_1 \leq c_1, j_2 \leq c_2 \}$ with $c_1 + c_2 = c$.
 Hence, it suffices to show that for every $c \in \ZZ$, there is $(c_1,c_2) \in \ZZ^2$ with $c_1 + c_2 = c$ such that $Q_{c_1,c_2} \subset \{ j \in \ZZ^2 \colon \langle j, \nu \rangle \leq \alpha \}$.
 This is the case if
 \[
  c_1 + \lambda c_2 \leq \alpha
  \quad
  \text{and}
  \quad
  c_1 + c_2 = c
 \]
 and a possible choice is $c_1 = c - \lceil (c - \alpha)/(1 - \lambda) \rceil$, $c_2 = c - c_1$, where $\lceil x \rceil$ denotes the least integer larger or equal than $x$.
 \end{proof}

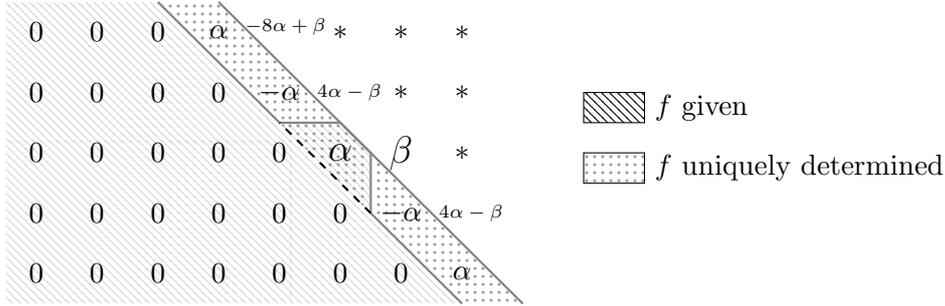
\begin{figure}
\caption{No unique continuation from a half space with border in a $45^\circ$ angle to the axes}
\label{fig:UC_Z_diagonal_half_space}
\begin{tikzpicture}[scale = .8]

\draw[pattern = north west lines] (9,2.5) rectangle (10,3);
\draw[anchor = west] (10,2.75) node {$f$ given};
\draw[pattern = dots, pattern color = black!40] (9,1.5) rectangle (10,2);
\draw[anchor = west] (10,1.75) node {$f$ uniquely determined};

\begin{scope}[yscale = -1, yshift = -4cm]

%  line separating half spaces
 \fill[pattern = north west lines, opacity = 0.4] (-.5,-.5) -- (2,-.5) -- (4,1.5) -- (5,1.5) -- (5.5,2) -- (5.5,3) -- (7,4.5)  -- (-.5,4.5) circle;
 \draw[thick, black!50]
 (2,-.5) -- (4,1.5) -- (5,1.5) -- (5.5,2) -- (5.5,3) -- (7,4.5);
 \fill[pattern = dots, pattern color = black!40] (2,-.5) -- (4,1.5) --
%  (5,1.5) -- (5.5,2) -- (5.5,3) --
 (7,4.5) -- (8,4.5) -- (6.5,3) --
%  (6.5,2) -- (6,1.5) -- (5,1.5) --
 (3,-.5) -- (2,-.5);
 \draw[thick, dashed] (5.5,3) -- (4,1.5);

  \draw[thick, black!50] (3,-.5) -- (8,4.5);

% Letters at the lattice sites

%  black dots
   \foreach \x in {0,...,2}{\draw (\x,0) node {$0$};}
   \foreach \x in {0,...,3}{\draw (\x,1) node {$0$};}
   \foreach \x in {0,...,4}{\draw (\x,2) node {$0$};}
   \foreach \x in {0,...,5}{\draw (\x,3) node {$0$};}
   \foreach \x in {0,...,6}{\draw (\x,4) node {$0$};}

%  row with $\alpha_1
\begin{scope}[xshift = 0cm]
   \draw (3,0) node {$\alpha$};
   \draw (4,1) node {$- \alpha$};
   \draw (5,2) node {\Large $\alpha$};
   \draw (6,3) node {$- \alpha$};
   \draw (7,4) node {$\alpha$};
\end{scope}

%  row with $0$
\begin{scope}[xshift = 1cm]
   \draw[] (3.1,-.1) node {\tiny $- 8 \alpha + \beta $};
   \draw[] (4.15,1) node {\tiny $4 \alpha - \beta$};
   \draw[] (5,2) node {\Large $\beta$};
   \draw[] (6.15,3) node {\tiny $4 \alpha - \beta$};
\end{scope}

%  row with $0$
\begin{scope}[xshift = 2cm]
   \draw[thick] (3,0) node {$\ast$};
   \draw[thick] (4,1) node {$\ast$};
   \draw[thick] (5,2) node {$\ast$};
\end{scope}

%  row with $\alpha_2$
\begin{scope}[xshift = 3cm]
   \draw (3,0) node {$\ast$};
   \draw (4,1) node {$\ast$};
\end{scope}

%  row with $0$
\begin{scope}[xshift = 4cm]
   \draw[thick] (3,0) node {$\ast$};
\end{scope}

%  row with $0$
\begin{scope}[xshift = 5cm]
\end{scope}

\end{scope}

\end{tikzpicture}
\\
Prescribing the value $\alpha$ in one point of the dotted strip completely determines $f$ on the dotted strip.
% Inductively, we end up with an infinite dimensional family of unique continuation from the diagonal half-space.
\end{figure}

\begin{example}[Inside-out continuation does not work on $\ZZ^d$]
 If $(- \Delta + V) f = 0$ on $\ZZ^d$ and $f = 0$ on a finite set $G \subset \ZZ^d$, we do not have $f = 0$ on $\ZZ^d$.
 In fact, $G$ is contained in a half-space the border of which is in a $45^\circ$ angle to a coordinate axis and even if $f$ vanished on the entire half space, we have seen that this cannot ensure a unique continuation.
\end{example}

\begin{example}[Outside-in continuation works on $\ZZ^d$]
 If however $(- \Delta + V) f = 0$ on $\ZZ^d$ and $f$ vanishes outside of a bounded set $G$, then $f$ vanishes on a half-space (with borders parallel to the axes) and therefore must vanish everywhere.
\end{example}

 So far, we encountered a couple of negative examples in which properties valid on $\RR^d$ do not hold any more on $\ZZ^d$.
 Nevertheless, outside-in unique continuation which holds on $\ZZ^d$ is sufficient to ensure continuity of the IDS of operators $- \Delta + V$ for ergodic $V: \ZZ^d \to \RR$ on the
Hilbert space $\ell^2(\ZZ^d) = \{ f: \ZZ^d \to \CC \mid \sum_{i \in \ZZ^d} |f(i)|^2 < \infty \}$, cf.~\cite{DelyonS-84}. Henceforth, when we speak about eigenfunctions, we always mean $\ell^2$-eigenfunctions. Let us explain their argument:
 \par
 Outside-in continuation implies that there are no finitely supported eigenfunctions.
 In fact, if this was not true, one could take a large box which contains the support of the eigenfunction.
 Outside the function is $0$, but by outside-in unique continuation, it follows that the function must be $0$ everywhere.
 By linearity, this implies that every eigenfunction of $- \Delta + V \mid_{\Lambda_L}$ with eigenvalue $E$ will be uniquely determined by its entries on $\partial_{(2)} \Lambda_L$, the set of sites in $\ZZ^d$ with distance at most $2$ to the complement of $\Lambda_L$.
 Now, continuity of the IDS at a point $E \in \RR$ is equivalent to the vanishing of
 \begin{equation}
 \label{eq:IDS_ZZ^d}
  \lim_{L \to \infty} \frac{1}{\lvert \Lambda_L \rvert}
  \sharp \{ \text{Eigenfunctions of}\ - \Delta + V \mid_{\Lambda_L} \ \text{with eigenvalue}\ E \},
 \end{equation}
 where $- \Delta + V \mid_{\Lambda_L}$ denotes the restriction of $- \Delta + V$ to $\{ j \in \ZZ^d \colon j \in \Lambda_L \}$ with \emph{simple boundary conditions}, i.e. the finite submatrix of $\{ \left\langle \delta_i, (- \Delta + V) \delta_j \right\rangle \}_{i,j \in \ZZ^d}$, corresponding to $i,j \in \Lambda_L \cap \ZZ^d$.
 By our considerations on unique continuation of eigenfunctions, the right hand side of Ineq.~\eqref{eq:IDS_ZZ^d} is bounded from above by
 \[
  \lim_{l \to \infty}
  \frac{\lvert \partial_{(2)} \Lambda_L \rvert}{\lvert \Lambda_L \rvert}
  =
  0.
 \]
 \par
 For the sake of completeness, we also mention that in 1981, Wegner showed Lipshitz continuity of the IDS and boundedness of the DOS for the usual Anderson model on $\ZZ^d$
 \[
  (H_\omega f)_i = (- \Delta f)_i + \omega_i \cdot f_i
  \quad
  i \in \ZZ^d
 \]
 in the case where the random variables $\omega_j$ are distributed according to a probability measure with a bounded density, cf.~\cite{Wegner-81}.
 Furthermore, with considerably more effort than in \cite{DelyonS-84}, Craig and Simon \cite{CraigS-83a}
established $\log$-H\"older continuity of the IDS if the potential $V \colon \ZZ^d \to \RR$ is a bounded, ergodic field.
This includes in particular the Anderson model with i.i.d. Bernoulli random variables.
 Finally, in \cite{CarmonaKM-87}, Thm. 2.2 it is shown that in dimension $d = 1$, the IDS for the Anderson model with Bernoulli random variables is not absolutely continuous, i.e. it does indeed inherit some irregularity from the random variables.

\section{Finitely supported eigenfunctions and the IDS on percolation graphs}
\label{sect:percolation}
We will now study site percolation on $\ZZ^d$.
Let $\{ q_j \}_{j \in \ZZ^d}$ be an i.i.d. collection of Bernoulli random variables on some probability space $(\Omega, \PP)$ with parameter $p \in (0,1)$, i.e.
\begin{align*}
 \PP ( q_j = 1) = p
 \quad
 \text{and}
 \quad
 \PP(q_j = 0) = 1-p.
\end{align*}
We call $X(\omega) := \{ j \in \ZZ^d: q_j = 1 \} \subset \ZZ^d$ the set of \emph{active} sites for the configuration $\omega \in \Omega$.
We say that $i,j \in X(\omega)$ are direct neighbours if they are direct neighbours in $\ZZ^d$.
$X(\omega)$ can be decomposed as a disjoint union of connected components, i.e. into subsets in which all sites are mutually joined by a path in $X(\omega)$ of direct neighbours.

The adjacency matrix $H_\omega$ on $X(\omega) $ is given by 
\[
 (H_\omega f)_i = \sum_{j \in X(\omega) \colon i \sim j} f_j.
\]
For a finite box $G \subset \ZZ^d$ let $H_{\omega,G}$ denote the restriction of $H_\omega$ to $G \cap X(\omega)$.
Then the finite volume normalized eigenvalue counting function on a box $\Lambda_L \subset \ZZ^d$ of side lenght $L$ is defined as
\[
 N_\omega^L(E) := \frac{ \sharp \{ \text{Eigenvalues}\ E_k\ \text{of}\ H_{\omega,\Lambda_L} \  \text{with}\ E_k \leq E \} }{\lvert \Lambda_L \rvert}.
\]
Similarly to the continuum case, one can thus define the integrated density of states $N(E)$ as a limit of finite volume normalized eigenvalue counting functions, at least on the points where $N(E)$ is continuous. We present here some results taken from \cite{Veselic-05b}.

\begin{theorem}[\cite{Veselic-05b}]
 There is $\Omega' \subset \Omega$ of full measure and a distribution function $N$ (the IDS of $H_\omega$)
such that for all $\omega \in \Omega'$ and all continuity points of $N$ we have
 \[
  \lim_{L \to \infty} N_\omega^L(E) = N(E).
 \]
\end{theorem}
%
% Now, in the light of the situation of the Anderson model, where the IDS is continuous, one can ask whether the restriction to continuity points of $N$ is a void restriction.
%
In contrast to the usual continuum Anderson model, the IDS for percolation graphs will be more irregular and have jumps.
This is due to the fact that $X(\omega)$ almost surely contains finite connected components on which the restriction of $H_\omega$ will carry $\ell^2(\ZZ^d)$-eigenfunctions of finite support.
Hence, if an eigenfunction is zero outside some large box, the box might still contain a finite component of $X(\omega)$ on which we non-zero eigenfunctions can be found.
Therefore, the outside-in unique continuation principle which had been used in the $\ZZ^d$ case to show continuity of the IDS, fails.
\begin{proposition}[\cite{ChayesCFST-86,Veselic-05b}]
 The set of discontinuity points of $N(E)$ is
 \[
   \cD= \{ E \in \RR \colon \exists \ \text{finite}\ G \subset \ZZ^d\ \text{and}\ f \in \ell^2(G)\ \text{such that}\ H^G f = E f \}
 \]
 which is an infinite subset of the algebraic numbers.
\end{proposition}
Now, one might wonder whether one can still expect some regularity of the IDS.
We start with a statement on the finite volume approximations.
\begin{theorem}[\cite{Veselic-05b}, Theorem 2.4]
 The normalized finite volume eigenvalue counting functions $N_\omega^L$ are right $\log$-H\"older continuous at $E \in \cD$
uniformly in $L$, i.e. for every $E \in \cD$ there is a constant $C_E$ such that for all $\epsilon \in (0,1)$, $L \in \NN$ and $\omega \in \Omega$ we have
 \[
  N_\omega^L(E + \epsilon) - N_\omega^L(E)
  \leq
  \frac{C_E}{\log(1 / \epsilon)}.
 \]
\end{theorem}
This immediately implies right $\log$-H\"older continuity of $N$ and is actually sufficient to ensure the convergence of normalized finite volume eigenvalue counting functions.
\begin{theorem}[\cite{Veselic-05b}, Corollary~2.5]
 The IDS $N$ is right $\log$-H\"older continuous and the convergence $\lim_{L \to \infty} N_\omega^L(E) = N(E)$ holds for all $E \in \RR$.
\end{theorem}
We conclude our comments on the regularity of the IDS of percolation Hamiltonians by examining the effect of adding a random potential.
Let
\[
 (V_\omega f)_i = \eta_i f_i,
 \quad
 i \in X(\omega)
\]
where $\{ \eta_j \}_{j \in \ZZ^d}$ is a process of positive, i.i.d. random variables independent of the percolation $\{q_j\}_{j \in \ZZ^d}$.
\begin{theorem}[\cite{Veselic-05b}, Theorem~2.6]
 If the probability measure corresponding to every $\eta_j$ has no atoms then the IDS of $H_\omega + V_\omega$ is continuous.
\end{theorem}
Most of the results of~\cite{Veselic-05b} hold for more general random operators defined on $\ell^2(G)$ where $G$ is a countable amenable group
(see also \cite{AntunovicV-08b}).
Furthermore, the pointwise convergence $\lim_{L \to \infty} N_\omega^L(E) = N(E)$ not only holds for all $E$, but actually uniformly in $E \in \RR$,
see~\cite{LenzV-09} and the references given there. One can even give an estimate on the approximation error in supremum norm, see \cite{SchumacherSV-16}.

\section{Existence of finitely supported eigenfunctions on planar graphs}
\label{sect:planar_graphs}

The graph Laplacian on $\ZZ^d$, defined in Definition \ref{def:Discrete_Lapliacian_ZZ^d}, has the following natural generalization to arbitrary graphs $G = (\cV,\cE)$ with vertex set $\cV$ and edge set $\cE$, with the only restriction of finite vertex degrees $|x| < \infty$ for all $x \in \cV$:
For a function $f: \cV \to \CC$, the \emph{(normalized) discrete Laplacian} is given by
\[
 \Delta_G f(x) = \frac{1}{|x|} \sum_{x \sim y} (f(x)-f(y)),
\]
where $x \sim y$ means that $x, y \in \cV$ are connected by an edge. The normalization by the vertex degree is just a scaling factor of the operator in the case of a regular graph (i.e., $|x|$ constant for all $x \in \cV$) such as $\ZZ^d$. For the rest of this note, we will use the normalized Laplacian.

A particular family of graphs are the \emph{planar graphs}, that is, graphs which have a realization in $\RR^2$ with non-crossing edges (edges can be curved and do not need to be straight lines).
For simplicity, we often identify planar realizations and their underlying discrete graphs.
The \emph{faces} of (a realization of) a planar graph $G$ are the closures of the connected components of the complement $\RR^2 \backslash G$.
We have already seen that the planar graph $\ZZ^2$ with edges between nearest neighbours does not admit finitely supported eigenfunctions (the faces of this graph are the unit squares $[k,k+1] \times [l,l+1]$ with
$k,l \in \ZZ$).
\par

A particular planar graph admitting finitely supported eigenfunctions is the \emph{Kagome lattice}. The Kagome lattice has attracted attention in the physics and mathematical physics community in connection with magnetic properties of certain crystal structures (see, e.g., \cite{Lawesetall-05,ChaudhuryCGLSWY-08}) and due to the emergence of butterfly spectra \cite{Hou-09,KerdelhueRL-14,HelfferKR-16}.

The Kagome lattice $K=(\cV,\cE)$ can be described as follows (see, e.g., \cite{LenzPPV-09}):  Let $w_1 = 1$ and $w_2 = e^{\pi i/3}$. Then the vertex set $\cV$ is given by the disjoint union
\[
 \cV
 =
 (2 \ZZ w_1 + 2 \ZZ w_2) \cup (w_1 + 2 \ZZ w_1 + 2 \ZZ w_2) \cup (w_2 + 2 \ZZ w_1 + 2 \ZZ w_2)
%  =:
%  V_0 \cup V_1 \cup V_2.
\]
A pair $x,y \in \cV$ is connected by a straight edge if and only if $|y-x| = 1$. The faces of this graph are regular triangles and hexagons, cf. Figure~\ref{fig:Kagome}.
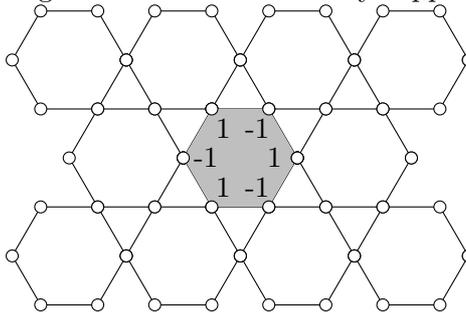
\begin{figure}
\label{fig:Kagome}
\caption{The Kagome lattice and a finitely supported eigenfunction}
 \begin{tikzpicture}[scale = 1.5]
  % Drawing the hexagons and the resulting grid
  \foreach \x in {0,...,3}{
   \begin{scope}[xshift = .5 cm + \x cm, yshift = -0.86602540378 * 1 cm]
    \draw (0:.5) -- (60:.5) -- (120:.5) -- (180:.5) -- (240:.5) -- (300:.5) -- (360:.5);
    \draw[fill = white] (0:.5) circle (1.5pt);
    \draw[fill = white] (60:.5) circle (1.5pt);
    \draw[fill = white] (120:.5) circle (1.5pt);
    \draw[fill = white] (180:.5) circle (1.5pt);
    \draw[fill = white] (240:.5) circle (1.5pt);
    \draw[fill = white] (300:.5) circle (1.5pt);
    \draw[fill = white] (360:.5) circle (1.5pt);
   \end{scope}
   }
  \foreach \x in {1,...,3}{
   \begin{scope}[xshift = {\x cm}, yshift = 0 cm]
    \draw (0:.5) -- (60:.5) -- (120:.5) -- (180:.5) -- (240:.5) -- (300:.5) -- (360:.5);
    \draw[fill = white] (0:.5) circle (1.5pt);
    \draw[fill = white] (60:.5) circle (1.5pt);
    \draw[fill = white] (120:.5) circle (1.5pt);
    \draw[fill = white] (180:.5) circle (1.5pt);
    \draw[fill = white] (240:.5) circle (1.5pt);
    \draw[fill = white] (300:.5) circle (1.5pt);
    \draw[fill = white] (360:.5) circle (1.5pt);
   \end{scope}
  }
  \foreach \x in {0,...,3}{
   \begin{scope}[xshift = .5 cm + \x cm, yshift = 0.86602540378 * 1 cm]
    \draw (0:.5) -- (60:.5) -- (120:.5) -- (180:.5) -- (240:.5) -- (300:.5) -- (360:.5);
    \draw[fill = white] (0:.5) circle (1.5pt);
    \draw[fill = white] (60:.5) circle (1.5pt);
    \draw[fill = white] (120:.5) circle (1.5pt);
    \draw[fill = white] (180:.5) circle (1.5pt);
    \draw[fill = white] (240:.5) circle (1.5pt);
    \draw[fill = white] (300:.5) circle (1.5pt);
    \draw[fill = white] (360:.5) circle (1.5pt);
   \end{scope}
   }

%   Drawing the particular cell

   \begin{scope}[xshift = 2 cm]
    \fill[black!25] (0:.5) -- (60:.5) -- (120:.5) -- (180:.5) -- (240:.5) -- (300:.5) -- (360:.5);
    \draw[fill = white] (0:.5) circle (1.5pt);
    \draw[fill = white] (60:.5) circle (1.5pt);
    \draw[fill = white] (120:.5) circle (1.5pt);
    \draw[fill = white] (180:.5) circle (1.5pt);
    \draw[fill = white] (240:.5) circle (1.5pt);
    \draw[fill = white] (300:.5) circle (1.5pt);
    \draw[fill = white] (360:.5) circle (1.5pt);

    \draw(0:.3) node {1};
    \draw (60:.3) node {{-1}};
    \draw (120:.3) node {1};
    \draw (180:.3) node {{-1}};
    \draw (240:.3) node {1};
    \draw (300:.3) node {{-1}};
   \end{scope}
\end{tikzpicture}
\end{figure}
It is easy to see that for a given hexagon
$$ H = \{x_0,x_1,\dots,x_5\} = \{ z_0 + e^{k\pi i/3} \mid k = 0,1,\dots,5 \} $$
with $z_0 \in (2 \ZZ + 1) w_1 + (2 \ZZ + 1) w_2$, the function
\begin{equation}
 \label{eq:Eigenfunction_Kagome}
  F_H(x)
  :=
  \begin{cases}
    0, & \text{if $x \in \cV \backslash H$,}\\
    (-1)^k, & \text{if $x \in H$,}
  \end{cases}
\end{equation}
satisfies $-\Delta_K F_H = 3/2 F_H$.
The following result tells us that, up to (infinite) linear combinations, these are the only $\ell^2$-eigenfunctions of the discrete Laplacian on the Kagome lattice:
\begin{proposition}[\cite{LenzPPV-09} Prop.~3.1]
\label{prop:Kagome1}
\begin{enumerate}[(a)]
 \item
 Let $F : \cV \to \CC$ be a finitely supported eigenfunction of $\Delta_K$.
 Then $- \Delta_K F = 3/2 F$ and $F$ is a linear combination of finitely many eigenfunctions $F_H$ of the above type~\eqref{eq:Eigenfunction_Kagome}.
 \item
 Let $H_i$, $i = 1, \dots,k$ be a collection of distinct, albeit not necessarily disjoint, hexagons.
 Then the set $F_{H_1}, \dots, F_{H_k}$ is linearily independent.
 \item
 If $g \in \ell^2(\cV)$ satisfies $- \Delta_K g = E g$, then $E = 3/2$.
 \item
 The space of $\ell^2(\cV)$-eigenfunctions to the eigenvalue $-3/2$ is spanned by finitely supported eigenfunctions.
\end{enumerate}
\end{proposition}

The next proposition shows that, similarly to the situation encountered in percolation, these finitely supported eigenfunctions give rise to a jump in the IDS.

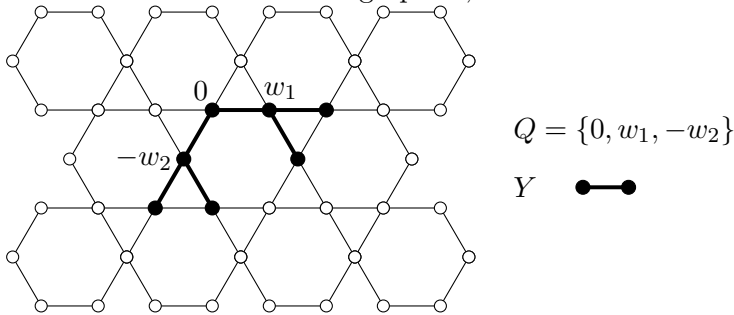
\begin{figure}
\label{fig:Kagome2}
\caption{Combinatorial fundamental domain $Q = \{0, w_1, -w_2 \}$ of the Kagome lattice and the metric subgraph $Y$, introduced in Section~\ref{sect:quantum_graphs}}
 \begin{tikzpicture}[scale = 1.5]
  % Drawing the hexagons and the resulting grid
  \foreach \x in {0,...,3}{
   \begin{scope}[xshift = .5 cm + \x cm, yshift = -0.86602540378 * 1 cm]
    \draw (0:.5) -- (60:.5) -- (120:.5) -- (180:.5) -- (240:.5) -- (300:.5) -- (360:.5);
    \draw[fill = white] (0:.5) circle (1.5pt);
    \draw[fill = white] (60:.5) circle (1.5pt);
    \draw[fill = white] (120:.5) circle (1.5pt);
    \draw[fill = white] (180:.5) circle (1.5pt);
    \draw[fill = white] (240:.5) circle (1.5pt);
    \draw[fill = white] (300:.5) circle (1.5pt);
    \draw[fill = white] (360:.5) circle (1.5pt);
   \end{scope}
   }
  \foreach \x in {1,...,3}{
   \begin{scope}[xshift = {\x cm}, yshift = 0 cm]
    \draw (0:.5) -- (60:.5) -- (120:.5) -- (180:.5) -- (240:.5) -- (300:.5) -- (360:.5);
    \draw[fill = white] (0:.5) circle (1.5pt);
    \draw[fill = white] (60:.5) circle (1.5pt);
    \draw[fill = white] (120:.5) circle (1.5pt);
    \draw[fill = white] (180:.5) circle (1.5pt);
    \draw[fill = white] (240:.5) circle (1.5pt);
    \draw[fill = white] (300:.5) circle (1.5pt);
    \draw[fill = white] (360:.5) circle (1.5pt);
   \end{scope}
  }
  \foreach \x in {0,...,3}{
   \begin{scope}[xshift = .5 cm + \x cm, yshift = 0.86602540378 * 1 cm]
    \draw (0:.5) -- (60:.5) -- (120:.5) -- (180:.5) -- (240:.5) -- (300:.5) -- (360:.5);
    \draw[fill = white] (0:.5) circle (1.5pt);
    \draw[fill = white] (60:.5) circle (1.5pt);
    \draw[fill = white] (120:.5) circle (1.5pt);
    \draw[fill = white] (180:.5) circle (1.5pt);
    \draw[fill = white] (240:.5) circle (1.5pt);
    \draw[fill = white] (300:.5) circle (1.5pt);
    \draw[fill = white] (360:.5) circle (1.5pt);
   \end{scope}
   }
%%%%%
%   Drawing the particular cell
%%%%%

\begin{scope}[xshift = 2 cm]
% Labels of the domain $Q = \{ 0, w_1, - w_2 \}$
   \draw (60:.7) node {$w_1$};
   \draw (120:.7) node {$0$};
   \draw (180:.85) node {$-w_2$};

% Domain $Y$
% 1) Vertices
   \draw[fill = black] (60:.5) circle (1.75pt);
   \draw[fill = black] (120:.5) circle (1.75pt);
   \draw[fill = black] (180:.5) circle (1.75pt);
   \draw[fill = black] (240:.5) circle (1.75pt);
   \begin{scope}[xshift = 1cm]
    \draw[fill = black] (120:.5) circle (1.75pt);
    \draw[fill = black] (180:.5) circle (1.75pt);
   \end{scope}
   \begin{scope}[xshift = -1cm]
    \draw[fill = black] (300:.5) circle (1.75pt);
   \end{scope}
% 2) Edges
   \draw[line width = 1.5pt] (0:.5) -- (60:.5) -- (120:.5) -- (180:.5) -- (240:.5);
   \begin{scope}[xshift = .5cm , yshift = 0.86602540378 * 1 cm]
    \draw[line width = 1.5pt] (240:.5) -- (300:.5);
   \end{scope}
   \begin{scope}[xshift = -1cm]
    \draw[line width = 1.5pt] (300:.5) -- (0:.5);
   \end{scope}
  \end{scope}

% Explanations
 \begin{scope}[xshift = 4cm]
  \draw[anchor = west] (0.3,.25) node {$Q = \{ 0, w_1, - w_2 \}$};
  \draw[fill = black] (1,-.25) circle (1.75pt);
  \draw[fill = black] (1.4,-.25) circle (1.75pt);
  \draw[line width = 1.5pt] (1,-.25) -- (1.4,-.25);
  \draw[anchor = west] (0.3,-.25) node {$Y$};
 \end{scope}
\end{tikzpicture}
\end{figure}

There is a $\ZZ^2$-action on the Kagome lattice via $T: \ZZ^2 \times \cV \to \cV$ via $T(\gamma,x) = T_\gamma(x) = 2 \gamma_1 w_1 + 2 \gamma_2 w_2 + x$ with combinatorial fundamental domain $Q = \{ 0, w_1, -w_2 \}$.
Any box $\Lambda_L \subset \ZZ^2$ gives rise to a set
\[
 \Lambda_{Q,L} := \bigcup_{\gamma \in \Lambda_L} T_\gamma(Q).
\]
Then Proposition~\ref{prop:Kagome1} has the following consequence:

\begin{proposition}[\cite{LenzPPV-09} Prop.3.3]
 \label{prop:Kagome2}
 Let $K$ be the Kagome lattice with the $\ZZ^2$ action introduced above.
 Then the IDS
 \[
  N(E)
  =
  \lim_{L \to \infty}
  \frac{1}{\lvert \Lambda_{Q,L} \rvert}
    \sharp \{ \text{Eigenfunctions of}\ - \Delta_K \mid_{\Lambda_{Q,L}} \ \text{with eigenvalue}\ \leq E \}
 \]
 exists and has the following properties:
 $N$ vanishes on $(- \infty, 0]$, is continuous on $\RR \setminus \{3/2\}$ and has a jump of size $1/3$ at $E = 3/2$.
 Moreover, $N$ is strictly monotone increasing on $[0,3/2]$ and $N(E) = 1$ for $E \geq 3/2$.
\end{proposition}

For the analysis of the IDS, in particular its jumps, an alternative formula is sometimes crucial
\[
  N(E)   =
  \frac{1}{\lvert Q \rvert}     \EE \left [\text{Tr} \, \chi_{Q} \, \chi_{(-\infty,E]}(\Delta_K) \right].
 \]
Here $\chi_{Q}$ denotes the multiplication operator with the indicator function of the fundamental cell $Q$,
whereas $\chi_{(-\infty,E]}(\Delta_K)$ is the spectral projector. Note that their product has finite trace.

An essential difference between the $\ZZ^2$-lattice and the Kagome lattice can be seen via a suitable notion of discrete curvature, defined on certain planar graphs called planar tessellations: A \emph{planar tessellation} $T = (\cV,\cE,\cF)$ is given by a realization of a planar graph with vertex set $\cV$, edge set $\cE$, and face set $\cF$, satisfying the following properties:
\begin{enumerate}[i)]
 \item
  Any edge is a side of precisely two different faces.
 \item
  Any two faces are disjoint or have precisely either a vertex or a side in common.
 \item
  Any face $f \in \cF$ is a polygon (i.e., homeomorphic to a closed disk) with finitely many sides, where $|f|$ denotes the number of sides.
 \item
  Every vertex $v \in \cV$ has finite degree $|v|$.
\end{enumerate}
We first define a curvature notion concentrated on the vertices. For this, we view the faces
adjacent to a vertex $v \in \cV$ as being represented by regular Euclidean polygons, that is, if $|f|=k$ its representation as regular $k$-gon has interior angles $(k-2)\pi/k$. The \emph{vertex curvature} $\kappa(v)$ in the vertex $v \in \cV$ is then defined via the angle defect/excess to $2\pi$ of the polygons around $v$:
$$ 2 \pi \kappa(v) = 2\pi - \sum_{f \ni v} \frac{|f|-2}{|f|} \pi = 2 \pi \left( 1 - \frac{|v|}{2} +
\sum_{f \ni v} \frac{1}{|f|} \right). $$
Unfortunately, this notion does not distinguish the Kagome lattice and the Euclidean lattice $\ZZ^2$, since both tessellations have vanishing vertex curvature. A finer curvature notion is defined on the corners (cf. \cite{BauesP-06}). A \emph{corner} of $T$ is a pair $(v,f) \in \cV \times \cF$ such that $v$ is a vertex of the polygon $f$.
The set of all corners of $T$ is denoted by $\cC = \cC(T)$. Then the \emph{corner curvature} of the corner $(v,f) \in \cC(T)$ is defined as
\[
 \kappa(v,f) := \frac{1}{|v|} + \frac{1}{|f|} - \frac{1}{2}.
\]
It is easy to see that we have
$$
\kappa(v) = \sum_{f \ni v} \kappa(v,f). $$
While $\ZZ^2$ has vanishing corner curvature in all corners, the Kagome lattice has corners with positive and negative corner curvature. There is the following general result:

\begin{theorem}[\cite{KlassertLPS-06}]
\label{thm:nonposcurv}
 Let $T = (\cV,\cE,\cF)$ be a planar tessellation with non-positive corner curvature, that is, $\kappa(v,f) \le 0$ for all $(v,f) \in \cC(T)$. Then $\Delta_T$ does not admit finitely supported eigenfunctions.
\end{theorem}

Note that Theorem~\ref{thm:nonposcurv} gives another proof of the fact that $\ZZ^2$ does not admit finitely supported eigenfunctions.

\begin{remark}
 In fact, Theorem \ref{thm:nonposcurv} holds for a much larger class of operators, called elliptic or nearest neighbour operators.
 Furthermore, it has been generalised to arbitrary connected, locally finite planar graphs in \cite{Keller-11} and to so-called polygonal complexes with planar substructures in \cite{KellerPP}.
\end{remark}

\section{Compactly supported eigenfunctions on quantum graphs}
\label{sect:quantum_graphs}
In this section, we introduce quantum graphs and study properties of the IDS in the particular example of the quantum graph associated to the Kagome lattice both in the equilateral and random setting. The results in the equilateral setting are based on the appearence of compactly supported eigenfunctions. The main reference for this section is \cite{LenzPPV-09}, providing further details. We start with some relevant definitions.

\begin{definition}
A \emph{metric graph} $(X,\ell)$ associated to a directed graph $G=(\cV,\cE)$ with maps $\partial_{\pm}: \cE \to \cV$ describing the direction of the edges (i.e., $\partial_-(e)$ is the source node and $\partial_+(e)$ the target node of the edge $e \in \cE$) consists of disjoint intervals $I_e = [0,\ell(e)]$ for each edge $e \in \cE$ which are identified at their end points in agreement with $G$ (for example, $0 \in I(e)$ is identified with $\ell(e') \in I_{e'}$ if $\partial_-(e) = \partial_+(e'))$.
The vertices and edges of $(X,\ell)$ are denoted by $\cV(X)$ and $\cE(X)$.
\end{definition}

Note that every metric graph $(X,\ell)$ is automatically also a metric space. The (one-dimensional) \emph{volume} of a metric subgraph $(X_0,\ell)$ of $(X,\ell)$ with a finite number of edges is defined as
\[
 \vol(X_0,\ell) = \sum_{e \in \cE(X_0)} \ell(e),
\]
and the \emph{boundary} $\partial X_0$ consists of all vertices of $X_0$ which are adjacent to vertices in $\cV(X)\backslash \cV(X_0)$.

Functions on a given metric graph $(X,\ell)$ are functions $f = \bigoplus_{e \in \cE} f_e$ with $f_e: I_e \to \CC$, and there is a natural Laplacian defined as follows:
\[
 \Delta_{X,\ell}  f = \bigoplus_{e \in \cE} f_e''.                                                                                                                                                                                                                                                                                              \]
A metric graph $(X,\ell)$ equipped with the Laplacian $\Delta_{X,\ell}$ is called a \emph{quantum graph}.

The relevant function spaces $C(X), L^2(X)$, and Sobolev spaces $H^{2,2}(X)$ are defined in a natural way (for details, see, e.g., \cite{LenzPPV-09}).
Note that for
\[
  H^{2,2}(X) \ni f = \bigoplus_{e \in \cE} f_e \in \bigoplus_{e \in \cE} H^{2,2}(I_e),
\]
the values $f_e(v), f_e'(v)$ for all $e \in \cE$ and $v \in \{\partial_\pm(e)\}$ are well defined.
To guarantee self-adjointness of the Laplacian, we assume a uniform positive lower bound on the edge lengths and assume appropriate vertex conditions for the functions $f_e$ at their end-points.
For simplicity, we only consider \emph{Kirchhoff vertex conditions} (other vertex conditions can be found, e.g., in \cite{LenzPPV-09}):
For all $v \in \cV$, we require
\begin{enumerate}[i)]
 \item $f_e(v) = f_{e'}(v)$ for all $e,e' \in \cE$ adjacent to $v$,
 \item $\sum_{\partial_+(e)=v} f_e'(v) = \sum_{\partial_-(e)=v} f_e'(v)$.
\end{enumerate}
Later, when we define the IDS via an exhaustion procedure, we will also need \emph{Dirichlet conditions} on certain vertices $v \in \cV$, which are defined by $f_e(v)=0$ for all $e \in \cE$ adjacent to $v$.
In this survey, we restrict our considerations to the Laplacian, but the results hold also in the more general setting of Schr\"odinger operators.

In the case of an \emph{equilateral} quantum graph, there is a well-known relation between the spectral components of the Laplacian $\Delta_{X,\ell}$ and the discrete graph Laplacian $\Delta_G$, by associating to a function $f \in H^{2,2}(X)$ with Kirchhoff boundary conditions the function $F \in \ell^2(\cV)$ via $F(v) = f(v)$:

\begin{proposition} (see, e.g., \cite{Cattaneo-97, Pankrashkin-06, Post-08})
Let $(X,\ell)$ with Kirchhoff Laplacian $\Delta_{X,\ell}$ be a quantum graph associated to the combinatorial graph $G=(\cV,\cE)$ with $l(e)=1$ for all $e \in \cE$ and $\Delta_G$ be the normalized discrete Laplacian. Then we have the following correspondence between the spectra:
\[
  E \in \sigma_\bullet(\Delta_{X,\ell})
  \quad
  \Longleftrightarrow
  \quad
  1-\cos(\sqrt{E}) \in \sigma_\bullet(\Delta_G)
\]
for all $E \not\in \Sigma^D = \{ (\pi k)^2 \mid k =1,2,\dots \}$, where $\bullet \in \{\emptyset,\mathrm{pp},\mathrm{disc},\mathrm{ess},\mathrm{ac},\mathrm{sc},\mathrm{p}\}$.
\end{proposition}

The values in $\Sigma^D$ above play a special role, since the quantum graph may have eigenfunctions $\Delta_{X,\ell} f = E f$ vanishing on all vertices (so-called \emph{Dirichlet eigenfunctions}).
They will appear as soon as the undirected underlying graph $G$ contains a cycle and must be of the form $(\pi k)^2$ for some $k = 1,2,\dots$. More precisely, the multiplicity of $(\pi k)^2$ is related to the global topology of the graph, as explained in \cite{LledoP-08}. Related
multiplicity calculations for quantum graphs were carried out in \cite{vonBelow-85}.

Note that the Kagome lattice, given in Figure \ref{fig:Kagome} as a subset of $\RR^2$, can be viewed as the corresponding metric graph $(X,\ell)$ with constant side length $\ell(e) = 1$ for all $e \in \cE$.
The map $T$ defined earlier can  be extended to $T: \ZZ^2 \times \RR^2 \to \RR^2$, $T_\gamma(x) = 2\gamma_1w_1+2\gamma_2w_2+x$, and induces a natural $\ZZ^2$-action on $(X,\ell)$ as a subset of $\RR^2$.
The closure of a fundamental domain of this $\ZZ^2$-action is given in Figure~6 %\ref{fig:Kagome2}
and is the induced metric subgraph
$(Y,\ell)$ with vertex set $\{ 0,w_1,2w_1,2w_1-w_2,-w_2,-2w_2,-2w_2+w_1\}$.
Any box $\Lambda_L \subset \ZZ^2$ gives rise to a metric subgraph $(\Lambda_{Y,L},\ell)$, defined as
\[
  \Lambda_{Y,L} := \bigcup_{\gamma \in \Lambda_L} T_\gamma(Y).
\]

Using the above spectral correspondence, it can be shown that Proposition \ref{prop:Kagome2} has the following analogue in the equilateral quantum graph on the Kagome lattice:

\begin{proposition}
 Let $(X,\ell)$ be the metric graph associated to the Kagome lattice $K=(\cV,\cE)$ with $\ell(e)=1$ for all $e \in \cE$. Then the IDS
 \[
  N(E)
  =
  \lim_{L \to \infty} \frac{1}{\vol(\Lambda_{Y,L})} \sharp \{ \text{Eigenfunctions of}\ - \Delta_{X,\ell} \mid_{\Lambda_{Y,L}} \ \text{with eigenvalue}\ \leq E \}
 \]
 exists, where $\Delta_{X,\ell} \mid_{X_0}$ is the restriction of $\Delta_{X,\ell}$ to the metric subgraph $(X_0,\ell)$ with Dirichlet vertex conditions on $\partial X_0$.
 Furthermore, all discontinuities of $N: \RR \to [0,\infty)$ are
\begin{enumerate}[i)]
 \item at $E = (2k+2/3)^2 \pi^2$, $k \in \ZZ$, with jumps of size $1/6$,
 \item at $E = k^2\pi^2$, $k \in \NN$, with jumps of size $1/2$.
\end{enumerate}
Moreover, $N$ is strictly increasing on the absolutely continuous spectrum of $\Delta_{X,\ell}$, which is explicitely given in \cite[Cor. 3.4]{LenzPPV-09}.
\end{proposition}

\begin{remark}
Note that there are two types of compactly supported eigenfunctions 
on a general equilateral quantum graph $(X,\ell)$ associated to a graph $G=(\cV,\cE)$:
\begin{enumerate}[i)]
\item eigenfunctions corresponding to finitely supported eigenfunctions of the discrete Laplacian $\Delta_G$,
\item Dirichlet eigenfunctions which appear as soon as the graph $G$ has cycles. For such a cycle of length $n$ in $G$, the corresponding cycle in the quantum graph $(X,\ell)$ can
be canonically identified with the interval $[0,n]$ where the end-points are identified, and any eigenfunction $\sin(k \pi)$ on $[0,n]$ gives rise to a corresponding Dirichlet eigenfunction with eigenvalue $k^2 \pi^2$.
Note that if $n$ is odd, $k \in \ZZ$ needs to be even.
\end{enumerate}
As a consequence, even though there are no jumps of the IDS of $\Delta_{\ZZ^d}$ in the discrete lattice $\ZZ^d$, jumps of the IDS of $\Delta_{X,\ell}$ appear in the equilateral quantum graph $(X,\ell)$ associated to $\ZZ^d$, due to the compactly supported eigenfunctions in ii), in dimension $d \ge 2$.
\end{remark}

%%%%%%
% The next figure is a sketch of a Dirichlet eigenfunction, but we are not sure if we really like it
%%%%%%

% \begin{figure}
%  \label{fig:Dirichlet_Eigenfunction}
%  \caption{A Dirichlet eigenfunction corresponding to the eigenvaule $4 \pi^2$ supported on a hexagonal circle of the equilateral Quantum graph associated to the Kagome lattice }
% \begin{tikzpicture}
% % Clipping the domain t make it look better
% \clip (-3.25,-1.25) rectangle (3.25,1.25);
%
% % Drawing the right and left hexagon
% \foreach \m in {-2,2}
% {
% \begin{scope}[xshift = \m cm]
%      \foreach \w in {0,60,120,180,240,300}{
%     \begin{scope}[rotate = \w]
% % Triangles
%      \draw[] (0:1) -- (30:{sqrt{3.}}) -- (60:1) -- (0:1);
% % Dots
%      \draw[fill = white] (0:1) circle (3pt);
%      \draw[fill = white] (60:1) circle (3pt);
%   \end{scope}
%     }
% \end{scope}
% }
%
% % Drawing the middle hexagon with the Dirichlet eigenfunction
%     \foreach \w in {0,60,120,180,240,300}{
%     \begin{scope}[rotate = \w]
% % Triangles
%      \draw[] (0:1) -- (30:{sqrt{3.}}) -- (60:1) -- (0:1);
% % sin waves
%     \begin{scope}[xshift = 1cm]
%      \begin{scope}[rotate = 120]
%       \draw[line width = 2pt, rounded corners = 5pt, dashed]
%        (0,0) --
%        (0.25,-.3) --
%        (0.5,.05) --
%        (0.75,.15) --
%        (1,0);
%      \end{scope}
%     \end{scope}
% % Dots
%      \draw[fill = white] (0:1) circle (3pt);
%      \draw[fill = white] (60:1) circle (3pt);
%   \end{scope}
%     }
% \end{tikzpicture}
% \end{figure}

Now we introduce randomness on the edge lengths of our metric Kagome lattice $(X,\ell)$.
Let $0 < \ell_{\min} < \ell_{\max} < \infty$ and $(\omega_e)_{e \in \cE}$ be a process of i.i.d. random variables on a probability space $(\Omega,\PP)$ with support in $[\ell_{\min},\ell_{\max}]$ and assume that every $\omega_e$ has a probability density $h \in C^1(\RR)$.
For every $\omega \in \Omega$, we consider the metric graph $(X,\ell_\omega)$, where $\ell_\omega(e) = \omega_e$ for all $e \in \cE$. This induces a random family of quantum graphs, called the \emph{random length model} associated to the Kagome lattice, consisting of $(X,\ell_\omega)_{\omega \in \Omega}$ with associated Laplacians $\Delta_{X,\ell_\omega}$.
Then the following Wegner estimate, linear in energy and volume, holds:
\begin{theorem}
 Let $(X,\ell_\omega)_{\omega \in \Omega}$ be the random length model associated to the Kagome lattice $K = (\cV,\cE)$ and $u > 1$. Then there exists a constant $C >0$, only depending on $u, \ell_{\min}, \ell_{\max}, \Vert h \Vert_\infty, \Vert h' \Vert_\infty$, such that, for all intervals $I \subset [1/u,u]$ and $L \in \NN$,
 \[
  \EE
  \left( \sharp \{ \text{Eigenfunctions of}\ - \Delta_{X,\ell_\omega} \mid_{\Lambda_{Y,L,\omega}} \ \text{with eigenvalue in } I \} \right)
  \leq
  C \cdot |I| \cdot |E(\Lambda_{Y,L})|,
 \]
where $(\Lambda_{Y,L,\omega},\ell_\omega)$ is a metric subgraph of $(X,\ell_\omega)$
defined analogously to the definition of $\Lambda_{Y,L}$ above.
\end{theorem}

A related Wegner estimate for the quantum graph associated to the lattice $\ZZ^d$ with random edge lengths and its application to localization was shown in \cite{KloppP-09}. The above Wegner estimate implies that randomness improves regularity of the IDS, as the next corollary states.

\begin{corollary}
 Let $(X,\ell_\omega)_{\omega \in \Omega}$ be the random length model associated to the Kagome lattice $K = (\cV,\cE)$.
 Then there is a unique function $N: \RR \to [0,\infty)$ such that for almost every $\omega \in \Omega$, the IDS corresponding to the quantum graph $(X,\ell_\omega,\Delta_{X,\ell_\omega})$ agrees with $N$.
 Moreover, $N$ is continuous on $\RR$ and even locally Lipschitz continuous on $(0,\infty)$.
\end{corollary}

\begin{remark}
 In fact, the result presented for the Kagome lattice holds in the much more general setting of a \emph{random length covering model}, as explained in \cite{LenzPPV-09}, where $\ZZ^2$ is replaced by a (not necessarily abelian) amenable group, acting cocompactly and isometrically on a connected, noncompact equilateral quantum graph and the boxes $\Lambda_L\subset \ZZ^2$ are replaced by a tempered F{\o}lner sequence.
\end{remark}

{\bf Acknowledgement:}
This work was partially financially supported by the Deutsche Forschungsgemeinschaft through the grants
VE 253/6-1 \emph{Unique continuation principles and equidistribution properties of eigenfunctions}
and
VE 253/7-1 \emph{Multiscale version of the Logvinenko-Sereda Theorem}.
While writing part of this article, NP and MT enjoyed the hospitality of the Isaac Newton Institute during the programme \emph{Non-Positive Curvature Group Actions and Cohomology},
supported by the EPSRC Grant EP/K032208/1.
We would like to thank Michela Egidi for reading a previous version of the manuscript.

%\bibliographystyle{alpha}
%\bibliography{lit_UCP-Geometry}

\newcommand{\etalchar}[1]{$^{#1}$}

\end{document}